\documentclass[
reprint,
amsmath,amssymb,
aps,
pra,
]{revtex4-2}

\usepackage{graphicx}
\usepackage{subfigure}
\usepackage{dcolumn}
\usepackage{bm}
\usepackage[ruled,vlined]{algorithm2e}
\usepackage{amsthm}
\newtheorem{theorem}{\indent Theorem}
\newtheorem{lemma}[theorem]{\indent Lemma}
\newtheorem{proposition}[theorem]{\indent Proposition}
\newtheorem{restate}{\indent Theorem}

\usepackage{hyperref}

\begin{document}

\preprint{APS/123-QED}

\title{Expressive Power and Limitations of Multi-photon Quantum Neural Networks}

\author{Zeyu Xiao}
\author{Weixu Shi}
\author{Yizhi Wang}\email{yizhiwang@nudt.edu.cn}
\author{Lingling Lao}
\author{Junjie Wu}\email{junjiewu@nudt.edu.cn}
\affiliation{College of Computer Science and Technology, National University of Defence Technology, Changsha 410073, China}

\date{\today}

\begin{abstract}
Quantum neural networks (QNNs) have shown promise in leveraging quantum computation for machine learning tasks. Utilizing multiple identical photons as input, multi-photon quantum neural networks (MPQNNs) have the potential to enhance the expressivity through increasing the photon number. However, how precisely the expressivity of an MPQNN is affected by an increase in photon number, and whether it can be infinitely enhanced by increasing the photon number, remains unexplored. In this work, we quantitatively estimate the expressivity of this model by deriving upper bounds on approximation error in two cases. In the case of a fixed observable, there exists a threshold that scales linearly with the mode number. Below the threshold, the expressivity of an MPQNN can be enhanced polynomially by increasing the photon number. Above the threshold, however, increasing the photon number does not affect the expressivity. In the case of a trainable observable,  the expressivity can always be enhanced polynomially by increasing the photon number. These findings are then validated by numerical simulations. Our work elucidates the performance enhancement of multi-photon quantum feature in QNNs, as well as its limitations, offering guidance for leveraging multi-photon advantages in quantum machine learning.
\end{abstract}

\maketitle

\section{Introduction}

Machine learning has achieved notable progress in classical computation, motivating efforts to explore quantum advantages in this field \cite{Cerezo2022}. Quantum neural networks (QNNs), as quantum counterparts of artificial neural networks, are considered promising candidates for integrating quantum computation with machine learning \cite{Abbas2021,Bharti2022,Cerezo2021}. It has been proved that uploading data multiple times in QNN enables universal approximation to a broad class of functions, the kind of which is called data re-uploading quantum neural networks (DRQNN)~\cite{Perez-Salinas2020,Perez-Salinas2021,Perez-Salinas2025,Schuld2021,Yu2022,Yu2024}. 
Despite its universal approximation property, the question of how quantum features enhance the neural network remains largely unexplored. 

Meanwhile, linear optical networks with multi-photon input have been shown quantum advantage both theoretically and experimentally in computational tasks, known as boson sampling \cite{Aaronson2011,Harrow2017,Zhong2020,Zhong2021,Madsen2022,Liu2026}. This have highlighted the potential of Fock states for quantum computation. Building on this insight, DRQNNs are generalized to multi-photon QNNs (MPQNNs) by extending the state space to the multi-photon Fock space, with numerical evidence showing that increasing the photon number enhances expressivity~\cite{Gan2022}. The DRQNN on a photonic integrated processor has been experimentally implemented \cite{Mauser2025}, demonstrating learning capability, but leaving the enhancement from multiple photons unexplored. It has been shown that the data consumption in learning stage can be reduced by multiple photons~\cite{Wang2026}. However, they did not adopt the data re-uploading scheme, which restricts the expressivity of the model.

MPQNN serves as a general model for the above research, but a thorough study on its expressivity is still missing. In this paper we tackle this problem by study the universality approximation of MPQNN. 
In \ref{sec:model}, we present the framework of MPQNN. 
In \ref{sec:description}, we explicitly derive the hypothesis space of MPQNNs, which serves as a crucial step to quantitatively characterize the expressivity. Then in  \ref{sec:fixed} and \ref{sec:trainable}, we establish approximation error bounds from the hypothesis space to a general class of functions. The error bound is related to the photon number $n$, the mode number $m$ and the layer number $L$. 
The error bounds take distinct forms depending on whether the observable is trainable or not, and leads to contrasting conclusions about the multiphoton-induced enhancement. 
In the trainable case, the expressivity of an MPQNN can always be enhanced by increasing the photon number, without altering the structure of the linear optical network. 
In the fixed case, surprisingly, once the photon number $n$ exceeds a threshold $m-2$, the additional photons cannot enhance further the expressivity of an MPQNN, complementing the findings in \cite{Wang2026}. The theoretical results are supported by numerical simulations in \ref{sec:experiment}. 
Our work provides a quantitative characterization of how the photon number can or cannot enhance the expressivity of quantum neural networks, offering a reference for the implementation of QNNs on photonic platforms.

\section{Model of MPQNNs}
\label{sec:model}

\begin{figure*}[t]
    \centering
    \includegraphics[width=0.9\linewidth]{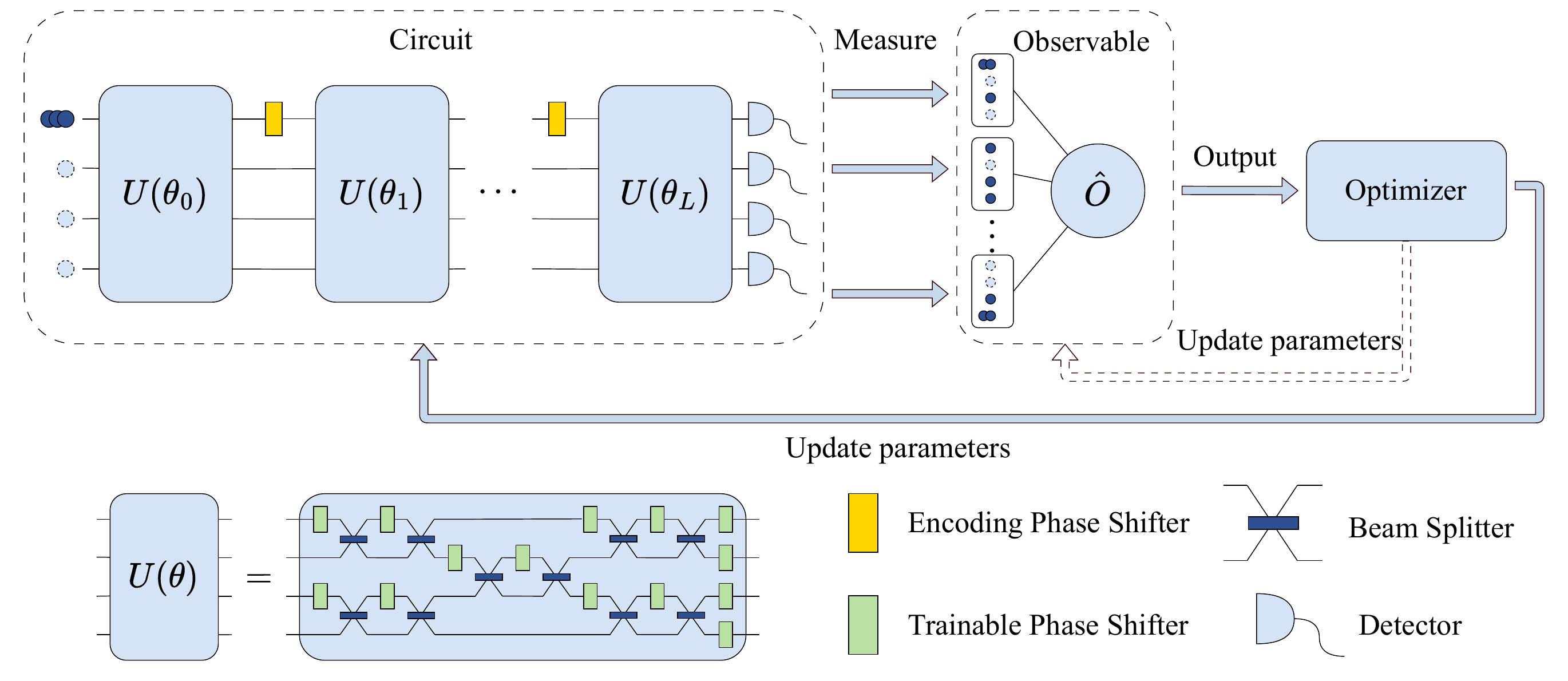}
    \caption{The model of MPQNNs. The model can be divided into the circuit part and the observable part. Phase shifters that serve as encoding blocks are depicted as yellow rectangles, while phase shifters with trainable parameters are depicted as green rectangles. The trainable blocks, depicted as blue rectangles, are linear optical networks with universal unitarity. After the photons pass the linear optical network, an observable is measured to obtain the output function. The measured observable can be fixed or trainable.}
    \label{fig:model}
\end{figure*}
Here we first revisit the fundamental definition of DRQNNs, then focus on their linear-optical-network implementation, and derive the mathematical description of the outputs of MPQNNs. The expressivity of an MPQNN is formulated as the universal approximation property of its output.

The different nature underlying quantum neural networks has led to a different way of data encoding. In classical neural networks, during the feed-forward process, each neuron has access to a copy of the outputs from all neurons in the preceding layer. However, data cannot be copied in quantum computers because of the no-cloning theorem. To overcome this limitation, DRQNNs take the approach of uploading data multiple times. Each of the $L$ layers of the quantum circuit is divided into the encoding block $U_E(x)$ and the trainable block $U_T(\theta_i)$, where $x$ denotes the input data and $\bm{\theta}=(\theta_0,\cdots,\theta_L)$ denotes the trainable parameters. The circuit can be represented as
\begin{eqnarray}
    U(x,\bm{\theta})=U_T(\theta_L)U_E(x)\cdots U_T(\theta_1)U_E(x)U_T(\theta_0).
\end{eqnarray}
An ordinary QNN can be treated as a special case of a DRQNN when $L=1$, as in \cite{Wang2023}.

The model of MPQNNs, as illustrated in Fig. \ref{fig:model}, can be divided into the circuit part and the observable part. In the circuit part, the quantum circuit implemented by the linear optical network possesses a data re-uploading structure. For the encoding blocks, with no loss of generality, we use phase shifters fixed in the first mode of the linear optical network, the unitary matrix of which is
\begin{eqnarray}
    U_{PS}(x)=\begin{pmatrix}
        e^{ix}&&&\\
        &1&&\\
        &&\ddots&\\
        &&&1
    \end{pmatrix}.
\end{eqnarray}
In this work, we focus on the task of univariate function approximation. For cases involving higher-dimensional input data, the individual components of the data can be encoded onto phase shifters on different modes, which is inherently scalable. The trainable blocks $U(\theta)$ are expected to be as universal as possible. Certain structures of the linear optical network constructed with beam splitters and phase shifters are known to have the flexibility to implement any unitary matrix by adjusting the parameters\cite{Reck1994,Clements2016,Guise2018}. For simplicity, we drop the trainable parameters $\theta_l$ from the expression and denote the unitary matrix of the trainable block in the $l$-th layer as $U_l$. The unitary matrix of a DRQNN on linear optics platform is
\begin{eqnarray}
    U(x)=U_LU_{PS}(x)\cdots U_1U_{PS}(x)U_0.
\end{eqnarray}

When we consider $n$ identical photons instead of one single photon, we use the (bosonic) Fock state to describe the state of photons. Denote the set of all possible photon partitions by
\begin{eqnarray}
    \Phi_{m,n}:=\{(s_1,\cdots,s_m):s_i\in\mathbb{N},\sum_{i=1}^ms_i=n\}.
\end{eqnarray}
The set $\Phi_{m,n}$ defines a basis of the Fock space, the dimension of which is $|\Phi_{m,n}|={n+m-1\choose m-1}$. Suppose that the input state is $\lvert\bm{s}\rangle=\lvert n,0,\cdots,0\rangle$, which means all $n$ photons are input on the first mode. A linear optical network, represented by its unitary matrix $U$, induces a unitary transformation $\hat{U}$ in the Fock space,
\begin{eqnarray}
    \hat{U}\lvert\bm{s}\rangle=\frac1{\sqrt{n!}}\left(\sum_{j=1}^mu_{j,1}\hat{a}_j^\dagger\right)^n\lvert\mathrm{vac}\rangle,
\end{eqnarray}
where $\hat{a}_k$ and $\hat{a}_k^\dagger$ are the annihilation and creation operators in mode $k$, and $\lvert\mathrm{vac}\rangle=\lvert0,\cdots,0\rangle$ is the vacuum state.

After passing through the linear optical network, the photons are measured, and the results are post-processed to obtain the observable expectation. The detectors obtain multiple distinct photon number partitions, which are granted different weights to realize the measurement of the expectation value of an observable
\begin{eqnarray}
    \hat{O}=\sum_{\bm{t}\in\Phi_{m,n}}o_{\bm{t}}\lvert\bm{t}\rangle\langle\bm{t}\rvert,
\end{eqnarray}
where $o_{\bm{t}}\in\mathbb{R}$ is the weight of the photon number partition $\bm{t}$. The output of an MPQNN can be represented as
\begin{eqnarray}
    h(x)=\langle\bm{s}\rvert\hat{U}^\dagger(x)\hat{O}\hat{U}(x)\lvert\bm{s}\rangle.
\end{eqnarray}
The value of the output function is fed into an optimizer. The optimizer can employ classical algorithms such as gradient descent to optimize a certain loss function, and update the trainable parameters of the linear optical network with optimized ones.

The weights in the measured observable $O_{\bm{t}}$ can be either fixed or trainable. In the fixed case, the weights in the observable are predetermined and remain unchanged during the training process. In the trainable case, however, the weights in the observable are optimized by the classical optimizer together with the trainable parameters in the linear optical network, and are updated during the training process. Intuitively, it would benefit for the expressivity of the MPQNN if the weights are allowed to be trainable rather than fixed. We separately analyzed the expressivity of the MPQNN under these two cases, thereby verifying this intuition.

The expressivity of an MPQNN can be characterized by the class of function that can be approximated by it. 
In an approximation task, the objective is to bring the output function of the MPQNN as close as possible to a target function by training the trainable parameters within the MPQNN. Let $\mathcal{H}$ be the family of all possible output functions of an MPQNN, referred to as the hypothesis space. For an arbitrary target function $f$, the approximation error is defined as $\inf_{h\in\mathcal{H}}\Vert f-h\Vert$, which is the distance from $f$ to the set $\mathcal{H}$. Here, $\Vert\cdot\Vert$ can be chosen as an appropriate norm. Our goal in the following sections is to establish an upper bound on the approximation error, given the photon number $n$, the mode number $m$, and the layer number $L$.

\section{Hypothesis Space of MPQNNs}
\label{sec:description}

Real-valued trigonometric polynomials plays an important role in most universal approximation theorems of DRQNNs\cite{Schuld2021, Yu2022,Neufeld2025,Tang2026}. The output function of a DRQNN is a trigonometric polynomial $p\in\mathbb{C}[e^{ix},e^{-ix}]$ with bounded degree, namely
\begin{eqnarray}
    p(x)=\sum_{\omega=-N}^Nc_\omega e^{i\omega x}.
\end{eqnarray}
As analyzed in \cite{Schuld2021}, the coefficients $c_\omega$ are determined by the trainable blocks within the DRQNN, while the spectrum is determined by the encoding blocks. Furthermore, the output function of a DRQNN should be real-valued, which implies that $c_{-\omega}=c_\omega^*$. Define $\mathcal{P}_N$ as the family of all real-valued trigonometric polynomials of degree $\le N$,
\begin{eqnarray}
    &&\mathcal{P}_N:=\{p\in\mathbb{C}[e^{ix},e^{-ix}]:\\\nonumber
    &&\deg(p)\le N;\forall x\in\mathbb{R},p(x)\in\mathbb{R}\}.
\end{eqnarray}
According to \cite{Schuld2021}, the hypothesis space of a DRQNN with $L$ layers is a subset of $\mathcal{P}_L$. When $L\to\infty$, the DRQNN possesses the universal approximation property.

The MPQNN is a natural generalization of the DRQNN to the Fock space, where the linear optical network has the same structure as the quantum circuit in the DRQNN, while the input state is a Fock state. Therefore, it is not difficult to understand that the hypothesis space of an MPQNN is a family of functions obtained by generalizing $\mathcal{P}_L$. When the photon number $n=1$, which is the single-photon case, the hypothesis space should be $\mathcal{P}_L$. When $n>1$, however, the multi-photon input increases the dimension of the Hilbert space, thereby enabling the MPQNN to express more complex functions. The following theorem characterizes all possible output functions of MPQNNs.

\begin{theorem}\label{thm:multi-photon}
    There exists an MPQNN whose output function is $h(x)$ if and only if there exists a real polynomial $g\in\mathbb{R}[x_1,\cdots,x_m]$ of total degree $\le n$ and trigonometric polynomials $y_1,\cdots,y_m\in\mathcal{P}_L$ such that $h=\pi_g(y_1,\cdots,y_m)$, where $\pi_g$ is defined as
    \begin{eqnarray}
        \pi_g(y_1,\cdots,y_m)(x):=g(y_1(x),\cdots,y_m(x)),
    \end{eqnarray}
    and $y_1,\cdots,y_m$ satisfy
    \begin{eqnarray}\label{eq:multi-photon}
        &y_j(x)\ge 0, &\forall j\in[m],\forall x\in\mathbb{R},\\
        &\sum_{j=1}^my_j(x)=1, &\forall x\in\mathbb{R}.\nonumber
    \end{eqnarray}
\end{theorem}

The proof of the theorem is detailed in Appendix \ref{apx:proof1}. In the expression of the output function of an MPQNN, the trigonometric polynomials $y_1,\cdots,y_m$ are determined by the trainable parameters in the linear optical network. 
Denote the family of all trigonometric polynomials satisfying such conditions by
\begin{eqnarray}
    &&\Delta:=\{(y_1,\cdots,y_m)\in\mathcal{P}_L^m:\\\nonumber
    &&\forall j\in[m],\forall x\in\mathbb{R},y_j(x)\ge 0;\forall x\in\mathbb{R},\sum_{j=1}^my_j(x)=1\}.
\end{eqnarray}
The range of $\Delta$ reflects the degrees of freedom of the quantum circuit within an MPQNN. On the other hand, the real polynomial $g$ is determined by the observable being measured. Under the action of $g$, the trigonometric polynomials $y_1,\cdots,y_m$ is mapped to a trigonometric polynomial of higher degree, reflecting the potential enhancement of the expressivity of an MPQNN by multi-photon input.

In an MPQNN, the measured observable can be either fixed or trainable, which leads to two distinct cases for the hypothesis space. In the case of a fixed observable, the detector reads the photon number partitions from the linear optical network, and then averages them with fixed weights to obtain the output of the MPQNN. In this case, $g$ is a fixed real polynomial, and the family of all possible output function of the MPQNN is $\pi_g(\Delta)$. However, when $g$ is fixed, $\pi_g(\Delta)$ is a bounded set, whereas the range of the target function $f$ is unbounded. This renders the distance from $f$ to $\pi_g(\Delta)$ unbounded, thereby preventing any quantitative estimation of the expressivity of the MPQNN. Referring to the standard practice in \cite{Yu2022,Yu2024,Tang2026}, we allow a resacle coefficient between the target function and the output function of the MPQNN. This yields the hypothesis space of the MPQNN in the case of a fixed observable,
\begin{eqnarray}
    \mathcal{H}_g:=\{\alpha\pi_g(\bm{y}):\bm{y}\in\Delta;\alpha\in\mathbb{R}\}.
\end{eqnarray}
In the case of a trainable observable, the output of the MPQNN is obtained by averaging over photon number partitions with trainable weights. These weights, like the trainable parameters within the linear optical network, is updated during the training process. This is equivalent to cascading a trainable linear layer after the MPQNN with a fixed observable, thereby introducing $n+m-1\choose m-1$ additional classical trainable parameters. Allowing a subset of parameters to be trainable while keeping another subset fixed may render the model more flexible. We leave this for future investigation. The hypothesis space of the MPQNN in the case of a trainable observable is
\begin{eqnarray}
    &&\mathcal{H}:=\{\pi_g(\bm{y}):\\\nonumber
    &&\bm{y}\in\Delta;g\in\mathbb{R}[x_1,\cdots,x_m];\deg(g)\le n\}.
\end{eqnarray}
Here, owing to the arbitrariness of the polynomial $g$, we no longer need to introduce a rescale coefficient. It is obvious that $\mathcal{H}_g\subseteq\mathcal{H}$ for all $g$. This enables MPQNNs with a trainable observable to overcome certain limitations of MPQNNs with a fixed observable.

\section{Expressivity of MPQNNs with Fixed Observable}
\label{sec:fixed}

In the case of a fixed observable, the measured observable $\hat{O}$ in the MPQNN is predetermined and remains unchanged during the training process. According to Theorem \ref{thm:multi-photon}, the fixed observable $\hat{O}$ yields a fixed real polynomial $g\in\mathbb{R}[x_1,\cdots,x_m]$ of total degree $\le n$. The hypothesis space of the MPQNN is $\mathcal{H}_g$. We quantitatively study the expressivity of the MPQNN by estimating the approximation error, i.e., the distance from the target function to the hypothesis space. Using Jackson's inequality \cite{Lorentz1966}, we prove an upper bound on the approximation error.

\begin{theorem}\label{thm:fixed}
    There exists a real polynomial $g\in\mathbb{R}[x_1,\cdots,x_m]$ of total degree $\le n$ and a constant $C_K>0$ for $K\in\mathbb{N}^+$ such that if $f\in C_{2\pi}(\mathbb{R})$ is a $K$-times continuously differentiable $2\pi$-periodic function, then
    \begin{eqnarray}\label{eq:error_bound}
        \inf_{h\in\mathcal{H}_g}\Vert f-h\Vert_\infty\le\frac{C_K\Vert f^{(K)}\Vert_\infty}{(dL)^K},
    \end{eqnarray}
    where $d=\min\{n,m-2\}$, and $f^{(K)}$ is the $K$-th derivative of the function $f$.
\end{theorem}

The proof of the theorem is detailed in Appendix \ref{apx:proof2}. The theorem is proven based on the fact that, for a suitably chosen $g$, any real-valued trigonometric polynomial of degree $\le dL$ can be expressed as the output function of the MPQNN, i.e., $\mathcal{P}_{dL}\subseteq\mathcal{H}_g$. As previously noted in Section \ref{sec:description}, the hypothesis space of a DRQNN is a subset of $\mathcal{P}_L$, since the spectrum of the output function is determined by the encoding blocks within the DRQNN. In contrast, the employment of multi-photon input state expands the hypothesis space of an MPQNN to $\mathcal{P}_{dL}$. The upper bound on the approximation error quantitatively describes the enhancement of the expressivity of MPQNNs brought about by increasing the photon number. For any target function, assuming the training process converges to the optimum, the error of the MPQNN in approximating the target function decays polynomially with $dL$. This implies that, provided the photon number $n$ does not exceed $m-2$, merely increasing the photon number can polynomially enhance the expressivity of an MPQNN without necessitating any alteration to the structure of the linear optical network.

On the other hand, the theorem also reveals the limitations of MPQNNs in terms of expressivity. The fact that $d=\min\{n,m-2\}$ indicates the existence of a threshold for the photon number $n$, which is $m-2$. When the photon number does not exceed the threshold, increasing the photon number enhances the expressivity of an MPQNN. Conversely, when the photon number exceeds the threshold, further increasing the photon number no longer affects the expressivity of the MPQNN. Note that, since the total degree of $g$ is no greater than $n$ and $y_1,\cdots,y_m\in\mathcal{P}_L$, the hypothesis space $\mathcal{H}_g$ is always a subspace of $\mathcal{P}_{nL}$. When $n$ does not exceed the threshold, $\mathcal{H}_g$ can cover the entire $\mathcal{P}_{nL}$. However, when $n$ exceeds the threshold, $\mathcal{H}_g$ becomes a low-dimensional subspace of $\mathcal{P}_{nL}$, rendering it impossible to cover the whole $\mathcal{P}_{nL}$. The existence of the photon number threshold is due to the fact that the number of trainable parameters in an MPQNN limits the degrees of freedom of the output function, thereby restricting the expressivity of the MPQNN. In the case of a fixed observable, all trainable parameters in an MPQNN reside within the trainable blocks of the linear optical network. The number of trainable parameters depends only on the mode number $m$ and the layer number $L$, and is independent of the photon number $n$. Therefore, if we increase the photon number without altering the structure of the linear optical network, the degrees of freedom of the output function are insufficient to allow it to express all functions in $\mathcal{P}_{nL}$.

\section{Expressivity of MPQNNs with Trainable Observable}
\label{sec:trainable}

Furthermore, we investigate the expressivity of MPQNNs in the case of a trainable observable. In this case, the measured observable $\hat{O}$ can be optimized during the training process in order to improve the performance of function approximation. That is, the weight $o_{\bm{t}}$ for each photon number partition in the expression of $\hat{O}$ is a trainable parameter, which is optimized synchronously with the trainable parameters in the linear optical network. As discussed in Section \ref{sec:description}, the hypothesis space of an MPQNN with a trainable observable is $\mathcal{H}$, which is the union of all hypothesis spaces $\mathcal{H}_g$ of MPQNNs with a fixed observable. It is not difficult to understand that an MPQNN with a trainable observable is more expressive than any MPQNN with a fixed observable. Consequently, the error bound in Theorem \ref{thm:fixed} also holds for the case of a trainable observable, which can also be derived from the fact that $\mathcal{H}_g\subseteq\mathcal{H}$ for all $g$. However, as we point out in Section \ref{sec:fixed}, MPQNNs with a fixed observable have limitations in expressivity. When the photon number exceeds a certain threshold, the error bound in Theorem \ref{thm:fixed} no longer decreases with further increases in the photon number, indicating that increasing the photon number beyond this point no longer affects the expressivity of the MPQNN. By proving an alternative upper bound on the approximation error, we demonstrate that MPQNNs with a trainable observable no longer suffer from such limitations. The expressivity of an MPQNN with a trainable observable can always be enhanced by increasing the photon number, without requiring any change to the structure of the linear optical network. We describe this result with a theorem.

\begin{theorem}\label{thm:trainable}
    For $K\in\mathbb{N}^+$, there exists a constant $C_K>0$ such that if $f\in C_{2\pi}(\mathbb{R})$ is a $K$-times continuously differentiable $2\pi$-periodic function, then
    \begin{eqnarray}
        \inf_{h\in\mathcal{H}}\Vert f-h\Vert_{\infty}\le\frac{C_K\lVert f^{(K)}\rVert_\infty}{d^K},
    \end{eqnarray}
    where $d=\min\{nL,\max\{(m-2)L,n\lfloor\frac{m-1}2\rfloor\}\}$.
\end{theorem}

The proof of the theorem is detailed in Appendix \ref{apx:proof3}. This error bound is obtained by combining the bound from Theorem \ref{thm:fixed} with another bound. As $n\to\infty$, we have $d\to\infty$, and the upper bound on the approximation error tends to zero, indicating that increasing the photon number can always enhance the expressivity of an MPQNN with a trainable observable. However, when the photon number $n$ and the layer number $L$ are within a certain range, the error bounds in Theorem \ref{thm:fixed} and Theorem \ref{thm:trainable} are both $O((nL)^{-K})$. In this case, an MPQNN with a trainable observable is not more expressive than an MPQNN with a suitably chosen fixed observable. Moreover, an MPQNN with a trainable observable achieves stronger expressivity at the cost of introduce $n+m-1\choose m-1$ additional trainable parameters, and this quantity increases exponentially with $n$. Whether the number of trainable parameters in the observable can be reduced without diminishing the expressivity of the MPQNN remains an open question.

\section{Numerical Experiments}
\label{sec:experiment}

\begin{figure*}[t]
    \subfigure[Approximation functions]{
        \includegraphics[width=0.45\linewidth]{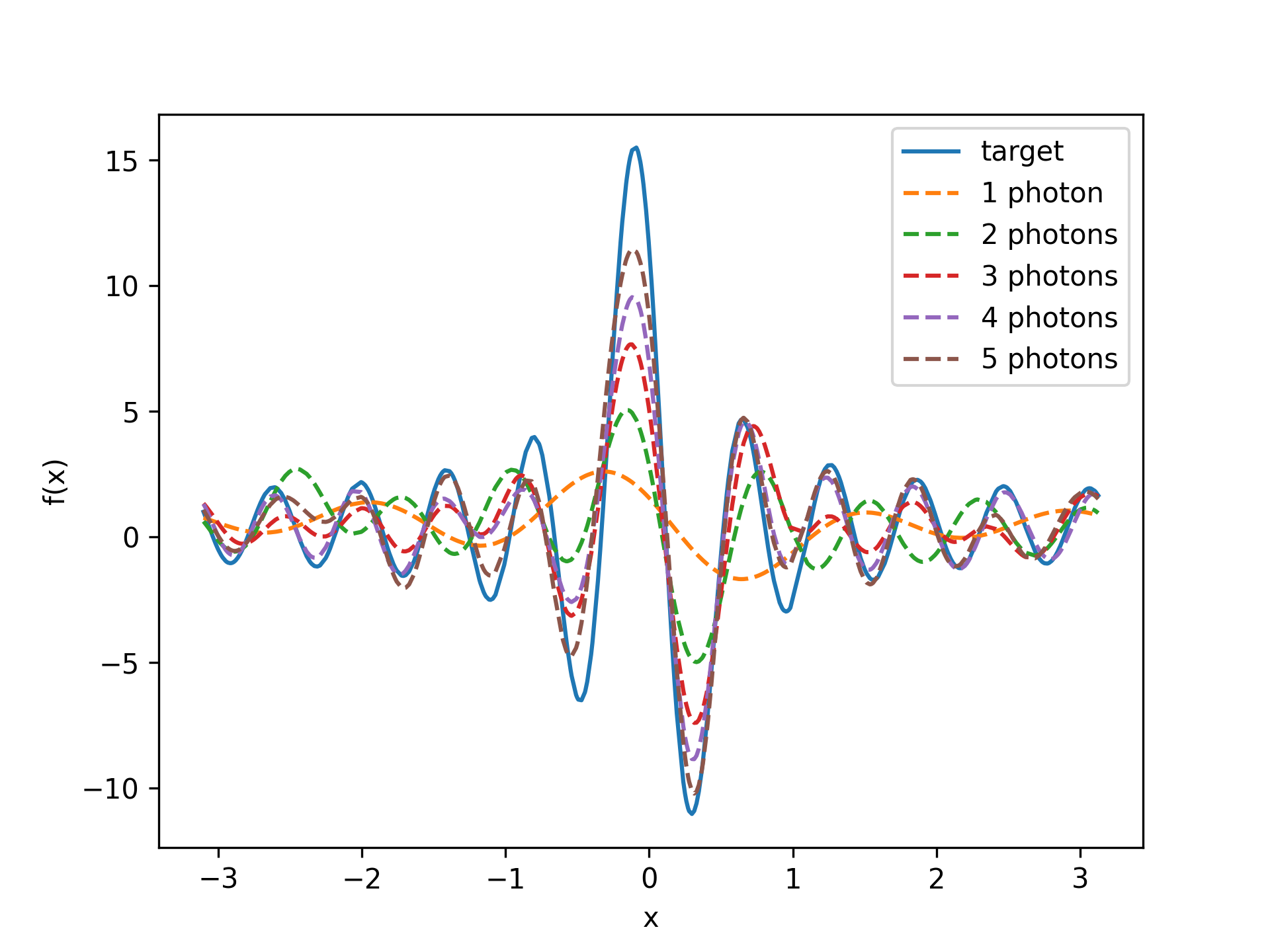}
    }
    \subfigure[Losses on the test dataset]{
        \includegraphics[width=0.45\linewidth]{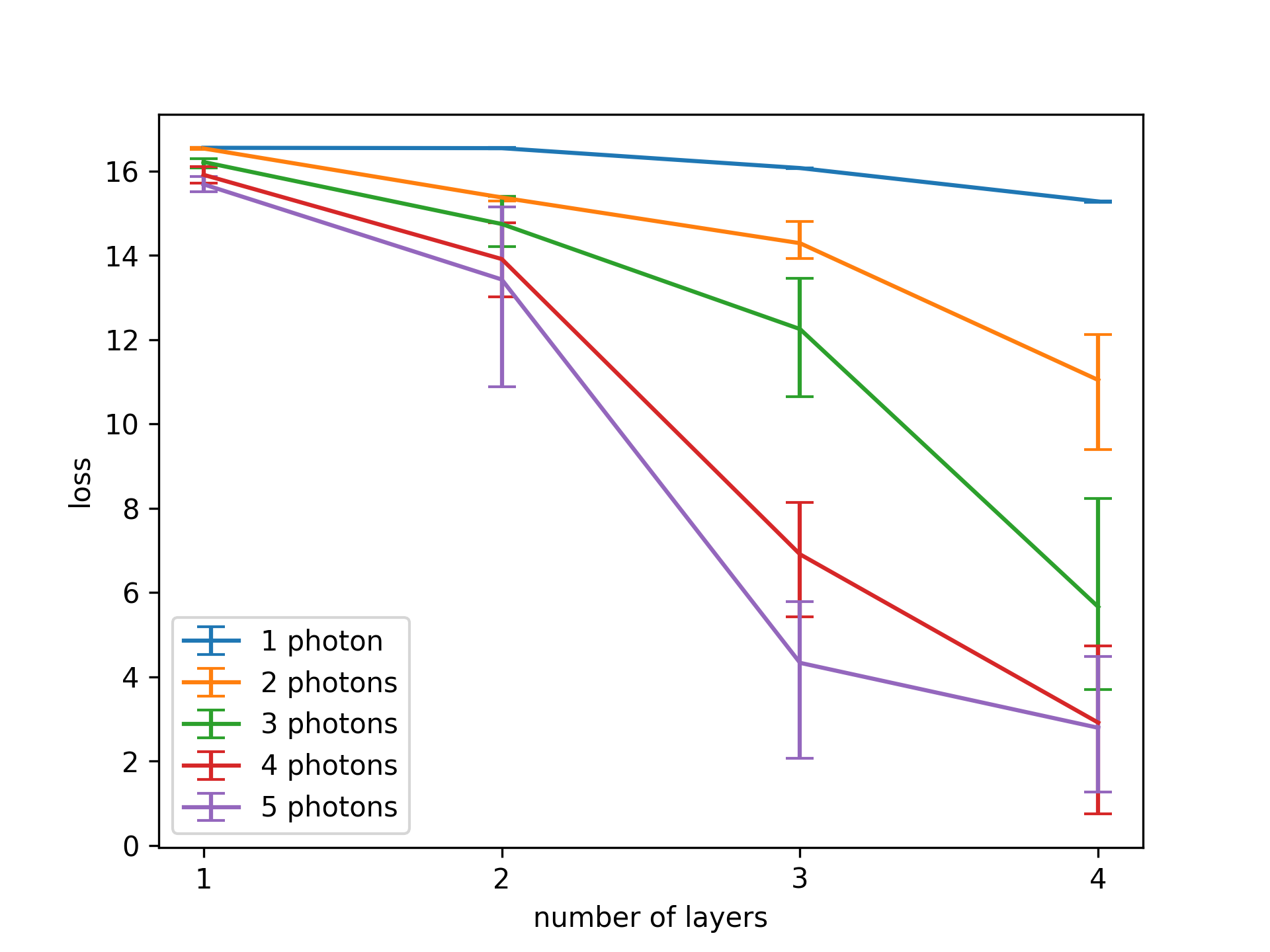}
    }
    \caption{(a) Approximation functions by an MPQNN with 4 layers. The photon number varies from 1 to 5. (b) Losses on the test dataset with the increase of the layer number, with variate photon number.}
    \label{fig:loss_and_approx}
\end{figure*}

We demonstrate the expressivity of MPQNNs by applying it in function approximation tasks. We simulate the training of MPQNNs using an algorithm based on dynamic programming. The error between the trained model and the target function indicates that the performance of the MPQNN improves with an increasing photon number.

The target function takes the form of a trigonometric polynomial
\begin{eqnarray}
    f(x)=\sum_{\omega=-N}^{N}c_{\omega}e^{i\omega x}.
\end{eqnarray}
In the simulation, we set the coefficients as follows. 
\begin{eqnarray}
    &&c_0=\frac12,\\
    &&c_\omega=\frac{|\omega|}{N}+i\frac{\omega}{N},\omega\neq0.
\end{eqnarray}
And we set the degree $N=10$. In the simulation, the mode number is fixed as $m=4$. We change the photon number $n$ from $1$ to $5$ and train MPQNNs with the maximum layer number being $L=4$. For each $n$, the input state $\lvert\bm{s}\rangle=\lvert n,0,\cdots,0\rangle$, which means that all $n$ photons are input into the first mode. We train the MPQNN with the mean square error (MSE) as the loss function, using the gradient descent algorithm to update the trainable parameters both in the linear optical network and the measured observable.

We propose a dynamic programming algorithm to compute the output photon number distribution of a linear optical network driven by Fock state inputs, and use it to simulate the training of MPQNNs. Compared with evaluating permanents via Ryser’s algorithm, our method offers an improvement in time complexity and substantially reduces runtime. Algorithmic details are given in the Appendix \ref{apx:algorithm}.

Fig.\ref{fig:loss_and_approx} (a) illustrates the function learned by the MPQNN with a different photon number. It can be demonstrated that models with more photons have more frequency components, thereby enabling them to approximate more complex target functions. The loss on the test dataset is depicted in Fig.\ref{fig:loss_and_approx} (b). The mean test loss decreases as the input photon number and the layer number increase, in agreement with the theoretic results.

\section{Conclusion}

In this work, we discussed the enhancement of QNNs using a Fock state as input. Following the data re-uploading construction of QNNs, we use phase shifters as encoding blocks and universal linear optical networks as trainable blocks. The output of an MPQNN is a measurement in the Fock space.

Theoretically, we analyzed the expressivity of MPQNNs. By realizing trigonometric polynomials, MPQNNs possesses the universal approximation property. To quantitatively describe the expressivity of MPQNNs, we estimated the approximation error for MPQNNs and investigated the relation between the approximation error and the photon number $n$, the mode number $m$ and the layer number $L$. In the case of a fixed observable, we derived an upper bound on the approximation error, which is $O((\min\{n,m-2\}L)^{-K})$. This indicates that there exists a threshold such that, when the photon number exceeds this threshold, further increasing the photon number cannot enhance the expressivity of an MPQNN with a fixed observable. In the case of a trainable observable, we also derived a error bound, which is $O((\min\{nL,\max\{(m-2)L,n\lfloor\frac{m-1}2\rfloor\}\})^{-K})$. This indicates that the expressivity of an MPQNN with a trainable observable can be enhanced by increasing the photon number without altering the structure of the linear optical network.

To efficiently simulate MPQNNs on the classical computer, we designed an algorithm based on dynamic programming to calculate the probability distribution of a linear optical network with multi-photon input. The result of the simulation demonstrated that the performance of the MPQNN increases as more photons are input.

For further research in the expressivity of MPQNNs, we expect an extension to the approximation of multivariate functions. And, as we have previously mentioned, the parameters in the measured observable can be partially fixed and partially trainable. How the number of such trainable parameters affects the expressivity of MPQNNs remains an open question worthy of further investigation. Except for the expressivity, the trainability and generalization of MPQNNs remain to be investigated. We believe that our work can provide theoretical assistance in demonstrating the advantages of multi-photon systems in quantum machine learning problems, and offer guidance for the application of QNNs on photonic platforms.

\appendix

\section{Proof of Theorem 1}
\label{apx:proof1}

In this appendix, we provide a proof of Theorem \ref{thm:multi-photon}, which characterizes the form of possible output functions of an MPQNN. We first present a lemma to describe the expressivity of MPQNNs with single-photon input. It is a natural generalization of Lemma 3 in \cite{Yu2022}, which is the special case where the mode number in the linear optical network $m=2$.

\begin{lemma}\label{thm:single-photon}
    There exist unitary matrices $U_0,\cdots,U_L$ such that the first column of the matrix
    \begin{eqnarray}
        U(x)=U_LU_{PS}(x)\cdots U_1U_{PS}(x)U_0.
    \end{eqnarray}
    is $u_1(x)=(u_{1,1}(x),\cdots,u_{m,1}(x))^T$ if and only if $u_{1,1},\cdots,u_{m,1}$ satisfy
    \begin{eqnarray}\label{eq:single-photon}
        &&\forall j\in[m],u_{j,1}\in\mathbb{C}[e^{ix}],\\
        &&\forall j\in[m],\deg(u_{j,1})\le L,\nonumber\\
        &&\forall x\in\mathbb{R},\sum_{j=1}^m|u_{j,1}(x)|^2=1.\nonumber
    \end{eqnarray}
\end{lemma}

\begin{proof}
    First, we prove that the elements of the first column of the matrix $U(x)$ satisfy the conditions in Eq.(\ref{eq:single-photon}), by induction on $L$. For the base case, if $L=0$, then $U(x)=U_0$ does not depend on $x$. Each element of the first column of $U_0$ can be viewed as a trigonometric polynomial of degree $0$, which satisfies the first and second conditions. The third condition follows from the unitarity of $U_0$.
    
    For the induction step, suppose that the forward direction of the lemma holds for $L-1$. Denote $U^\prime(x)=U_{L-1}U_{PS}(x)\cdots U_1U_{PS}(x)U_0$, the first column of which is $w_1(x)=(w_{1,1}(x),\cdots,w_{m,1}(x))^T$. According to the inductive hypothesis, $w_{1,1}(x),\cdots,w_{m,1}(x)$ satisfy the three conditions. Denote
    \begin{eqnarray}
        U_L=\begin{pmatrix}
            v_{1,1}&\cdots&v_{1,m}\\
            \vdots&\ddots&\vdots\\
            v_{m,1}&\cdots&v_{m,m}
        \end{pmatrix}.
    \end{eqnarray}
    Then the $j$-th element of the first column of $U(x)$ is
    \begin{eqnarray}
        u_{j,1}(x)=v_{j,1}e^{ix}w_{1,1}(x)+\sum_{k=2}^mv_{j,k}w_{k,1}(x).
    \end{eqnarray}
    Since $w_{1,1},\cdots,w_{m,1}\in\mathbb{C}[e^{ix}]$ are of degree at most $L-1$, $u_{1,1},\cdots,u_{m,1}$ satisfy the first and second conditions. The third condition follows from the unitarity of $U(x)$.
    
    Now we prove by induction on $L$ that for all $u_{1,1},\cdots,u_{m,1}$ satisfying the conditions, there exist unitary matrices $U_0,\cdots,U_L$ that make the first column of $U(x)$ be $u_1(x)=(u_{1,1}(x),\cdots,u_{m,1}(x))^T$. For the base case, if $L=0$, then $U(x)=U_0$ and $u_1=(u_{1,1},\cdots,u_{m,1})^T$ do not depend on $x$. Applying Gram-Schmidt orthogonalization to $u_1$ yields the unitary matrix $U_0$.
    
    For the induction step, suppose that the reverse direction of the lemma holds for $L-1$ and adopt the same notation as above. Giving $U(x)$, the elements of the first column of which satisfy the three conditions, we aim to construct the unitary matrices $U_0,\cdots,U_L$. Let $d=\max_{1\le j\le m}\{\deg(u_{j,1})\}$ and without loss of generality $\deg(u_{1,1})=d$. 
    To see this, assume that one of the trigonometric polynomials $u_{2,1},\cdots,u_{m,1}$ is of degree $d$, say $\deg(u_{2,1})=d$. Denote
    \begin{eqnarray}
        u_{1,1}(x)=\sum_{k=0}^da_ke^{ikx},\\
        u_{2,1}(x)=\sum_{k=0}^db_ke^{ikx}.
    \end{eqnarray}
    We can reduce the degree of $u_{2,1}$ by left-multiplying $U_L$ with a unitary matrix.
    \begin{eqnarray}
        &&\begin{pmatrix}
            \cos\theta&-e^{i\phi}\sin\theta&&&\\
            \sin\theta&e^{i\phi}\cos\theta&&&\\
            &&1&&\\
            &&&\ddots&\\
            &&&&1
        \end{pmatrix}
        \begin{pmatrix}
            \sum_{k=0}^da_ke^{ikx}\\
            \sum_{k=0}^db_ke^{ikx}\\
            \vdots
        \end{pmatrix}\\\nonumber
        &&=\begin{pmatrix}
            \sum_{k=0}^d(a_k\cos\theta-b_ke^{i\phi}\sin\theta)e^{ikx}\\
            \sum_{k=0}^d(a_k\sin\theta+b_ke^{i\phi}\cos\theta)e^{ikx}\\
            \vdots
        \end{pmatrix}.
    \end{eqnarray}
    For $a_d,b_d\in\mathbb{C}$, there always exist $\theta,\phi\in\mathbb{R}$ that make $a_d\sin\theta+b_de^{i\phi}\cos\theta=0$. Thus, without loss of generality, we may assume that $\deg(u_{j,1})\le d-1,j=2,\cdots,m$. 
    
    Since $u_{1,1}(x),\cdots,u_{m,1}(x)$ satisfy the third condition, all coefficients of $\sum_{j=1}^m|u_{j,1}(x)|^2$ except the constant term must be zero. By calculating the coefficient of the $d$-th degree term, we can find that $a_da_0=0$. By the definition of $d$, $a_d\neq 0$. Then $a_0=0$, which means there exists a trigonometric polynomial $w_{1,1}\in\mathbb{C}[e^{ix}]$ such that $u_{1,1}(x)=e^{ix}w_{1,1}(x)$. Let $w_{j,1}(x)=u_{j,1}(x),j=2,\cdots,m$. It is clear that $\deg(w_{j,1})\le d-1\le L-1,j=1,\cdots,m$, and that $\sum_{j=1}^m|w_{j,1}(x)|^2=\sum_{j=1}^m|u_{j,1}(x)|^2=1$. According to the inductive hypothesis, there exist $U_0,\cdots,U_{L-1}$ such that the first column of $U_{L-1}U_{PS}(x)\cdots U_1U_{PS}(x)U_0$ is $(w_{1,1}(x),\cdots,w_{m,1}(x))^T$. Let $U_L$ be the identity matrix. Then, the first column of $U_LU_{PS}(x)\cdots U_1U_{PS}(x)U_0$ is $(u_{1,1}(x),\cdots,u_{m,1}(x))^T$.
\end{proof}

The lemma indicates that the elements of the unitary matrix of the linear optical network of an MPQNN are trigonometric polynomials that satisfy the normalization condition. And the degrees of the trigonometric polynomials are upper-bounded by the layer number $L$. Based on this lemma, we can prove Theorem \ref{thm:multi-photon}, which we restate for completeness.

\begin{restate}
    There exists an MPQNN whose output function is $h(x)$ if and only if there exists a real polynomial $g\in\mathbb{R}[x_1,\cdots,x_m]$ of total degree $\le n$ and trigonometric polynomials $y_1,\cdots,y_m\in\mathcal{P}_L$ such that $h=\pi_g(y_1,\cdots,y_m)$, where $y_1,\cdots,y_m$ satisfy
    \begin{eqnarray} \label{eq:thm1_constr}
        &&\forall j\in[m],\forall x\in\mathbb{R},y_j(x)\ge 0,\\
        &&\forall x\in\mathbb{R},\sum_{j=1}^my_j(x)=1.\nonumber
    \end{eqnarray}
    And $\pi_g$ is the assignment map
    \begin{eqnarray}
        \pi_g(y_1,\cdots,y_m)(x):=g(y_1(x),\cdots,y_m(x)).
    \end{eqnarray}
\end{restate}

\begin{proof}
    We first prove that the output function of an MPQNN necessarily takes the form described by the theorem. The output function of an MPQNN is
    \begin{eqnarray}
        h(x)&&=\langle\bm{s}\rvert\hat{U}^\dagger(x)\hat{O}\hat{U}(x)\lvert\bm{s}\rangle\\
        &&=\langle\bm{s}\rvert\hat{U}^\dagger(x)\left(\sum_{\bm{t}\in\Phi_{m,n}}o_{\bm{t}}\lvert\bm{t}\rangle\langle\bm{t}\rvert\right)\hat{U}(x)\lvert\bm{s}\rangle\\
        &&=\sum_{\bm{t}\in\Phi_{m,n}}o_{\bm{t}}|\langle\bm{t}\rvert\hat{U}(x)\lvert\bm{s}\rangle|^2\\
        &&=\sum_{\bm{t}\in\Phi_{m,n}}\frac{o_{\bm{t}}}{n!}\left|\langle\bm{t}\rvert\left(\sum_{j=1}^mu_{j,1}(x)\hat{a}_j^\dagger\right)^n\lvert\mathrm{vac}\rangle\right|^2\\
        &&=\sum_{\bm{t}\in\Phi_{m,n}}\frac{o_{\bm{t}}}{n!}\left|\frac{n!}{\sqrt{t_1!\cdots t_m!}}\prod_{j=1}^mu_{j,1}^{t_j}(x)\right|^2\\
        &&=\sum_{\bm{t}\in\Phi_{m,n}}\frac{n!}{t_1!\cdots t_m!}o_{\bm{t}}\prod_{j=1}^m|u_{j,1}(x)|^{2t_j}.
    \end{eqnarray}
    Let $y_j(x)=|u_{j,1}(x)|^2$ for $j\in[m]$. From Lemma \ref{thm:single-photon}, it follows that $(y_1,\cdots,y_m) \in \Delta$. Let
    \begin{eqnarray}
        g(x_1,\cdots,x_m)=\sum_{\bm{t}\in\Phi_{m,n}}\frac{n!}{t_1!\cdots t_m!}o_{\bm{t}}\prod_{j=1}^mx_j^{t_j}.
    \end{eqnarray}
    Clearly, $g$ is a real polynomial of total degree $\le n$. Thus, $h(x)=\pi_g(y_1,\cdots,y_m)$ takes the form described by the theorem.

    Now we prove that for all $h$ taking the form, there exists an MPQNN whose output function is $h$. Suppose that $(y_1,\cdots,y_m)\in \Delta$. For $j\in[m]$, according to the Fejer-Riesz theorem \cite{Riesz1990}, there exists $u_{j,1}\in\mathbb{C}[e^{ix}]$ such that $y_j(x)=|u_{j,1}(x)|^2$ for all $x\in\mathbb{R}$. Thus, $u_{1,1},\cdots,u_{m,1}$ satisfy the conditions in Eq.(\ref{eq:single-photon}). From Lemma \ref{thm:single-photon}, it follows that there exists unitary matrices $U_0,\cdots,U_L$ such that the first column of $U_LU_{PS}(x)\cdots U_1U_{PS}(x)U_0$ is $(u_{1,1}(x),\cdots,u_{m,1}(x))^T$. Suppose that $g\in\mathbb{R}[x_1,\cdots,x_m]$ is of total degree $\le n$. We can always construct a homogeneous real polynomial $g^\prime$ of total degree $n$ by multiplying the lower degree terms in $g$ by $1=x_1+\cdots+x_m$ considering \ref{eq:thm1_constr}. Then $g^\prime$ takes the form
    \begin{eqnarray}
        g^\prime(x_1,\cdots,x_m)=\sum_{\bm{t}\in\Phi_{m,n}}g_{\bm{t}}\prod_{j=1}^mx_j^{t_j}.
    \end{eqnarray}
    By the construction of $g^\prime$, it follows that
    \begin{eqnarray}
        \pi_{g^\prime}(y_1,\cdots,y_m)=\pi_g(y_1,\cdots,y_m)=h.
    \end{eqnarray}
    Let
    \begin{eqnarray}
        \hat{O}=\sum_{\bm{t}\in\Phi_{m,n}}\frac{t1!\cdots t_m!}{n!}g'_{\bm{t}}\lvert\bm{t}\rangle\langle\bm{t}\rvert.
    \end{eqnarray}
    The output function of the MPQNN defined by trainable blocks $U_0,\cdots,U_L$ and observable $\hat{O}$ is $h$.
\end{proof}

From the proof of the theorem, it can be seen that the trigonometric polynomials $y_1,\cdots,y_m$ are determined by the trainable blocks within the MPQNN, while the real polynomial $g$ is determined by the measured observable.

\section{Proof of Theorem 2}
\label{apx:proof2}
In this appendix, we provide a proof of Theorem \ref{thm:fixed}, which gives an upper bound on the approximation error of an MPQNN with a fixed observable. An important tool we employ to prove the error bound is Jackson's inequality.

\begin{proposition}[\cite{Lorentz1966}, Theorem 4.3, p. 57]\label{thm:jackson}
    For $K\in\mathbb{N}$, there exists a constant $C_K>0$ such that if $f\in C_{2\pi}(\mathbb{R})$ is a $K$-times continuously differentiable $2\pi$-periodic function, then for every $N\in\mathbb{N}^+$ there exists a trigonometric polynomial $p_f\in\mathcal{P}_N$ such that
    \begin{eqnarray}
        \Vert f-p_f\Vert_\infty\le\frac{C_K\omega(f^{(K)},\frac{1}{N})}{N^K},
    \end{eqnarray}
    where $\omega(f,\delta)$ is the modulus of continuity,
    \begin{eqnarray}
        \omega(f,\delta):=\sup_{x,t\in\mathbb{R},|t|<\delta}|f(x+t)-f(x)|.
    \end{eqnarray}
    Moreover, if $f$ is $(K+1)$-times continuously differentiable, then
    \begin{eqnarray}
        \Vert f-p_f\Vert_\infty\le\frac{C_K\Vert f^{(K+1)}\Vert_\infty}{N^{K+1}}.
    \end{eqnarray}
\end{proposition}

Given the existence of $p_f$, we can decompose the approximation error into two parts using the triangle inequality.
\begin{eqnarray}
    \inf_{h\in\mathcal{H}_g}\Vert f-h\Vert_\infty\le\Vert f-p_f\Vert_\infty+\inf_{h\in\mathcal{H}_g}\Vert p_f-h\Vert_\infty.
\end{eqnarray}
The first term on the right-hand side can be bounded using Jackson's inequality. As for the second term, it is either strictly equal to zero or unbounded.
\begin{eqnarray}
    \sup_{p\in\mathcal{P}_N}\inf_{h\in\mathcal{H}_g}\Vert p-h\Vert_\infty=
    \begin{cases}
        0,&\text{if }\mathcal{P}_N\subseteq\mathcal{H}_g,\\
        +\infty,&\text{otherwise}.
    \end{cases}
\end{eqnarray}
This can be understood by noting that $\mathcal{H}$ is a closed cone. By definition, $\mathcal{P}_N$ is a $(2N+1)$-dimensional $\mathbb{R}$-linear space, and $\Delta$ is a closed subset of $\mathcal{P}_L$. Since $g$ is a real polynomial of total degree $\le n$, $\pi_g(\Delta)$ is a closed subset of $\mathcal{P}_{nL}$. Then $\mathcal{H}_g$ is a closed cone in $\mathcal{P}_{nL}$. In order to prove Theorem \ref{thm:fixed}, we present a sufficient condition for the existence of $g$ such that $\mathcal{H}_g$ covers $\mathcal{P}_N$.

\begin{lemma}\label{thm:fixed_dimension}
    There exists a real polynomial $g\in\mathbb{R}[x_1,\cdots,x_m]$ of total degree $\le n$ such that $\mathcal{P}_{dL}\subseteq\mathcal{H}_g$, where $d=\min\{n,m-2\}$.
\end{lemma}

\begin{proof}
    We prove the lemma in two steps. First, we construct a relatively interior point $\bm{y}$ of $\Delta$ and a real polynomial $g$ of degree $\le n$ satisfying
    \begin{eqnarray}
        &&\pi_g(\bm{y})=\bm{0},\\
        &&D\pi_g(\bm{y})(T)=\mathcal{P}_{dL},
    \end{eqnarray}
    where $\bm{0}$ is the zero function, $D\pi_g(\bm{y})$ is the Frechet derivative of $\pi_g$ at $\bm{y}$, and
    \begin{eqnarray}
        T:=\{(v_1,\cdots,v_m)\in\mathcal{P}_L^m:\sum_{j=1}^mv_j=\bm{0}\}.
    \end{eqnarray}
    Then, we utilize the properties satisfied by $g$ and $\bm{y}$ to prove that $\mathcal{P}_{dL}\subseteq\mathcal{H}_g$.

    Construct $\bm{y}=(y_1,\cdots,y_m)$ as
    \begin{eqnarray}
        &&y_1(x)=\cdots=y_{m-1}(x)=\frac{1}{m}+\epsilon\cos(Lx),\\
        &&y_m(x)=\frac{1}{m}-\epsilon(m-1)\cos(Lx).
    \end{eqnarray}
    It can be seen that $\bm{y}\in\mathcal{P}_L^m$ by noting that $\cos(Lx)=\frac{1}{2}(e^{iLx}+e^{-iLx})$. By choosing a sufficiently small $\epsilon$, we can ensure that $\bm{y}$ is a relatively interior point of $\Delta$. Construct $g$ as
    \begin{eqnarray}
        g(x_1,\cdots,x_m)=\sum_{j=1}^d(q_j(z_j)-q_j(z_{d+1})),
    \end{eqnarray}
    where $z_j=\frac{1}{\epsilon}(x_j-\frac{1}{m})$ for $j\in[d+1]$. And $q_j\in\mathbb{R}[x]$ is an antiderivative of the Chebyshev polynomial of degree $j-1$, i.e.
    \begin{eqnarray}
        q_j^\prime(\cos x)=\cos((j-1)x).
    \end{eqnarray}
    Since $d=\min\{n,m-2\}$, $g$ is a real polynomial of total degree $\le n$. For brevity, denote
    \begin{eqnarray}
        g_j:=\frac{\partial}{\partial x_j}g.
    \end{eqnarray}
    Differentiating $\pi_g(\bm{y})(x)$ with respect to $x$ yields
    \begin{eqnarray}
        \frac{d}{dx}\pi_g(\bm{y})(x)&&=\sum_{j=1}^m\pi_{g_j}(\bm{y})\frac{d}{dx}y_j\\
        &&=\sum_{j=1}^{d+1}\pi_{g_j}(\bm{y})\frac{d}{dx}y_j\\
        &&=-L\sin(Lx)\left(\sum_{j=1}^{d+1}\pi_{g_j}(\bm{y})\right)\\
        &&=0.
    \end{eqnarray}
    This indicates that $\pi_g(\bm{y})$ is a constant function. By adjusting the constant term of $q_j$, we can make $\pi_g(\bm{y})=\bm{0}$. By the definition of $T$, it follows that
    \begin{eqnarray}
        D\pi_g(\bm{y})(T)=&&\sum_{j=1}^{m-1}(\pi_{g_j}(\bm{y})-\pi_{g_m}(\bm{y}))\mathcal{P}_L\\
        =&&\sum_{j=1}^{d+1}\pi_{g_j}(\bm{y})\mathcal{P}_L\\
        =&&\sum_{j=1}^dq_j^\prime(\cos(Lx))\mathcal{P}_L+\\\nonumber
        &&\left(-\sum_{j=1}^dq_j^\prime(\cos(Lx))\right)\mathcal{P}_L\\
        =&&\sum_{j=1}^dq_j^\prime(\cos(Lx))\mathcal{P}_L\\
        =&&\sum_{j=1}^d\cos((j-1)Lx)\mathcal{P}_L\\
        =&&\mathcal{P}_{dL}.
    \end{eqnarray}
    Thus, the construction of $g$ and $\bm{y}$ satisfies the conditions.

    Given the construction of $g$ and $\bm{y}$, we can prove that $\mathcal{P}_{dL}\subseteq\mathcal{H}_g$. Since $\bm{y}$ is a relatively interior point of $\Delta$, $T$ is the tangent space at $\bm{y}$ to $\Delta$. By the definition of the interior point, there exists $\delta>0$ such that the open ball in $T$
    \begin{eqnarray}
        B_T(\bm{y},\delta):=\{\bm{z}\in T:\Vert\bm{z}-\bm{y}\Vert_\infty<\delta\}
    \end{eqnarray}
    is contained in $\Delta$. Define a map
    \begin{eqnarray}
        F:B_T(\bm{0},\delta)\to\mathcal{P}_{dL},F(\bm{v})=\pi_g(\bm{y}+\bm{v}).
    \end{eqnarray}
    It follows that $F(\bm{0})=\pi_g(\bm{y})=\bm{0}$ and $DF(\bm{0})=D\pi_g(\bm{y})|_T$ is a surjection onto $\mathcal{P}_{dL}$. According to Theorem 1 in \cite{Sussmann2003}, $F$ is an open at $\bm{0}$. That is, there exists $\rho>0$ such that
    \begin{eqnarray}
        B_{\mathcal{P}_{dL}}(\bm{0},\rho)\subseteq F(B_T(\bm{0},\delta))\subseteq\pi_g(\Delta).
    \end{eqnarray}
    For every $p\in\mathcal{P}_{dL}$, there exists $\alpha>0$ such that $\Vert\alpha p\Vert_\infty<\rho$, leading to $\alpha p\in\pi_g(\Delta)$. It follows that $\mathcal{P}_{dL}\subseteq\mathcal{H}_g$.
\end{proof}

In fact, by comparing dimensions, we can also obtain a necessary condition for the existence of $g$ such that $\mathcal{H}_g$ covers $\mathcal{P}_N$. A straightforward condition is that $N\le nL$, since $\mathcal{H}_g\subseteq\mathcal{P}_{nL}$. Define a map
\begin{eqnarray}
    G:\mathbb{R}\times\mathcal{P}_L^m\to\mathcal{P}_{nL},G(\alpha,\bm{y})=\alpha\pi_g(\bm{y}).
\end{eqnarray}
By definition, $\mathcal{H}_g=G(\mathbb{R}\times\Delta)$. Note that $G$ is a smooth map from $\mathbb{R}\times\mathcal{P}_L^m$ to $\mathcal{P}_{nL}$, and that $\mathbb{R}\times\Delta$ is an $((m-1)(2L+1)+1)$-dimensional manifold with boundary. According to the corollary of Sard's theorem \cite{Lee2012}, the dimension of $\mathcal{H}$ is at most $(m-1)(2L+1)+1$. Recalling that the dimension of $\mathcal{P}_N$ is $2N+1$, a necessary condition for $\mathcal{P}_N\subseteq\mathcal{H}_g$ is
\begin{eqnarray}
    (m-1)(2L+1)+1\ge 2N+1.
\end{eqnarray}
Solving the inequality and noting that $N$ is an integer, we obtain
\begin{eqnarray}
    N\le(m-1)L+\left\lfloor\frac{m-1}{2}\right\rfloor.
\end{eqnarray}
Combining with $N\le nL$, we give an upper bound on $N$ for the existence of $g$ such that $\mathcal{P}_N\subseteq\mathcal{H}_g$,
\begin{eqnarray}
    N\le\min\left\{nL,(m-1)L+\left\lfloor\frac{m-1}{2}\right\rfloor\right\},
\end{eqnarray}
which is asymptotically the same as $\min\{n,m-2\}L$. Although not used in the proof of the error bound, this result implies limitations on the expressivity of MPQNNs. Now we move on to prove Theorem \ref{thm:fixed}, which we restate for completeness.

\begin{restate}
    There exists a real polynomial $g\in\mathbb{R}[x_1,\cdots,x_m]$ of total degree $\le n$ and a constant $C_K>0$ for $K\in\mathbb{N}^+$ such that if $f\in C_{2\pi}(\mathbb{R})$ is a $K$-times continuously differentiable $2\pi$-periodic function, then
    \begin{eqnarray}\label{eq:error_bound}
        \inf_{h\in\mathcal{H}_g}\Vert f-h\Vert_\infty\le\frac{C_K\Vert f^{(K)}\Vert_\infty}{(dL)^K},
    \end{eqnarray}
    where $d=\min\{n,m-2\}$.
\end{restate}

\begin{proof}
    According to Lemma \ref{thm:fixed_dimension}, There exists a real polynomial $g\in\mathbb{R}[x_1,\cdots,x_m]$ of total degree $\le n$ such that $\mathcal{P}_{dL}\subseteq\mathcal{H}_g$. And according to Proposition \ref{thm:jackson}, there exists a constant $C_K>0$ and $p_f\in\mathcal{P}_{dL}$ such that
    \begin{eqnarray}
        \Vert f-p_f\Vert_\infty\le\frac{C_K\Vert f^{(K)}\Vert_\infty}{(dL)^K}.
    \end{eqnarray}
    Using the triangle inequality, the approximation error can be decomposed as
    \begin{eqnarray}
        \inf_{h\in\mathcal{H}_g}\Vert f-h\Vert_\infty\le\Vert f-p_f\Vert_\infty+\inf_{h\in\mathcal{H}_g}\Vert p_f-h\Vert_\infty.
    \end{eqnarray}
    Since $\mathcal{P}_{dL}\subseteq\mathcal{H}_g$, the second term of the right-hand side of the inequality equals to zero. Then,
    \begin{eqnarray}
        \inf_{h\in\mathcal{H}_g}\Vert f-h\Vert_\infty&&\le\Vert f-p_f\Vert_\infty\\
        &&\le\frac{C_K\Vert f^{(K)}\Vert_\infty}{(dL)^K}.
    \end{eqnarray}
    Thus, the approximation error can be upper bounded.
\end{proof}

\section{Proof of Theorem 3}
\label{apx:proof3}

In this appendix, we provide a proof of Theorem \ref{thm:trainable}, which gives an upper bound on the approximation error of an MPQNN with a trainable observable. Our approach is the same as in Appendix \ref{apx:proof2}: we prove that $\mathcal{P}_N\subseteq\mathcal{H}$ for some $N$, and then use Jackson's inequality to obtain an error bound. We prove a different error bound from that in Theorem \ref{thm:fixed}, and combine them to obtain the error bound in Theorem \ref{thm:trainable}. First, we present a lemma describing an upper bound on $N$ for $\mathcal{P}_N\subseteq\mathcal{H}$.

\begin{lemma}\label{thm:trainable_dimension}
    $\mathcal{P}_{nd}\subseteq\mathcal{H}$, where $d=\min\{L,\lfloor\frac{m-1}2\rfloor\}$.
\end{lemma}

\begin{proof}
    For $k\in\mathbb{N}^+$, noting that $\cos(kx)=\frac{1}{2}(e^{ikx}+e^{-ikx})$ and $\sin(kx)=\frac{i}{2}(e^{ikx}-e^{-ikx})$, we have $\cos(kx),\sin(kx)\in\mathcal{P}_k$. Let $\bm{z}\in\mathcal{P}_d^{2d}$ be
    \begin{eqnarray}
        \bm{z}(x)=(\cos{x},\sin{x},\cdots,\cos{dx},\sin{dx}).
    \end{eqnarray}
    The orbit $\{\bm{z}(x):x\in\mathbb{R}\}$ is a compact set in $\mathbb{R}^{2d}$. There exists a $2d$-dimensional convex set
    \begin{eqnarray}
        C=\rm{conv}\{\bm{v_1},\cdots,\bm{v}_{2d+1}\}\subseteq\mathbb{R}^{2d}
    \end{eqnarray}
    that contains the orbit. For every $x\in\mathbb{R}$, $\bm{z}(x)$ can be represented as a convex combination of $\bm{v}_1,\cdots,\bm{v}_{2d+1}$. Denote
    \begin{eqnarray}
        \bm{z}(x)=\sum_{j=1}^{2d+1}\lambda_j(x)\bm{v}_j.
    \end{eqnarray}
    Now we construct $\bm{y}\in\mathcal{P}_L^m$ as
    \begin{eqnarray}
        y_j(x)=\begin{cases}
            \lambda_j(x),&j=1,\cdots,2d+1,\\
            0,&j=2d+2,\cdots,m.
        \end{cases}
    \end{eqnarray}
    By convexity, it follows that $\bm{y}\in\Delta$ and
    \begin{eqnarray}
        \forall x\in\mathbb{R},\bm{z}(x)=\sum_{j=1}^my_j(x)\bm{v}_j.
    \end{eqnarray}
    For any $p\in\mathcal{P}_{nd}$, there exists a real polynomial $q\in\mathbb{R}[x_1,\cdots,x_{2d}]$ of total degree $\le n$ such that
    \begin{eqnarray}
        p=\pi_q(z_1,\cdots,z_{2d}).
    \end{eqnarray}
    Substitution gives
    \begin{eqnarray}
        p=g(y_1,\cdots,y_m),
    \end{eqnarray}
    where $g\in\mathbb{R}[x_1,\cdots,x_m]$ is another real polynomial of total degree $\le n$. Thus, we have $\mathcal{P}_{nd}\subseteq\mathcal{H}$.
\end{proof}

Using Jackson's inequality, we can now readily prove Theorem \ref{thm:trainable}, which we restate for completeness.

\begin{restate}
    For $K\in\mathbb{N}^+$, there exists a constant $C_K>0$ such that if $f\in C_{2\pi}(\mathbb{R})$ is a $K$-times continuously differentiable $2\pi$-periodic function, then
    \begin{eqnarray}
        \inf_{h\in\mathcal{H}}\Vert f-h\Vert_{\infty}\le\frac{C_K\Vert f^{(K)}\Vert_\infty}{d^K},
    \end{eqnarray}
    where $d=\min\{nL,\max\{(m-2)L,n\lfloor\frac{m-1}2\rfloor\}\}$.
\end{restate}

\begin{proof}
    Let $d_1=\min\{nL,n\lfloor\frac{m-1}2\rfloor\}$. According to Proposition \ref{thm:jackson}, there exists a constant $C_{K,1}>0$ and $p_f\in\mathcal{P}_{d_1}$ such that
    \begin{eqnarray}
        \Vert f-p_f\Vert_\infty\le\frac{C_{K,1}\Vert f^{(K)}\Vert_\infty}{d_1^K}.
    \end{eqnarray}
    According to Lemma \ref{thm:trainable_dimension}, $\mathcal{P}_{d_1}\subseteq\mathcal{H}$. Using the decomposition of approximation error in Appendix \ref{apx:proof2}, we can obtain
    \begin{eqnarray}
        \inf_{h\in\mathcal{H}}\Vert f-h\Vert_{\infty}\le\frac{C_{K,1}\Vert f^{(K)}\Vert_\infty}{d_1^K}.
    \end{eqnarray}
    According to Theorem \ref{thm:fixed}, noting that $\mathcal{H}_g\subseteq\mathcal{H}$ for all $g$, there exists a constant $C_{K,2}>0$ such that
    \begin{eqnarray}
        \inf_{h\in\mathcal{H}}\Vert f-h\Vert_{\infty}\le\frac{C_{K,2}\Vert f^{(K)}\Vert_\infty}{d_2^K},
    \end{eqnarray}
    where $d_2=\min\{nL,(m-2)L\}$. Let $C_K=\max\{C_{K,1},C_{K,2}\}$ and
    \begin{eqnarray}
        d&&=\max\{d_1,d_2\}\\
        &&=\min\{nL,\max\{(m-2)L,n\lfloor\frac{m-1}2\rfloor\}\}.
    \end{eqnarray}
    Then we can obtain
    \begin{eqnarray}
        \inf_{h\in\mathcal{H}}\Vert f-h\Vert_{\infty}\le\frac{C_K\Vert f^{(K)}\Vert_\infty}{d^K},
    \end{eqnarray}
    which completes the proof.
\end{proof}

\section{Simulation Algorithm}
\label{apx:algorithm}

We introduce the algorithm used in our numerical simulations and present the results of the simulations. To describe the simulation algorithm in detail, we first recall the connection between the amplitude of Fock states and the permanent of matrix. We compute the output of an MPQNN by explicitly compute the final state. Suppose that the input photon number partition is $\bm{s}\in\Phi_{m,n}$, and the initial state $\lvert\bm{s}\rangle$ undergoes a unitary transformation $\hat{U}$ induced by a unitary matrix $U$. The amplitude of a partition $\bm{t}\in\Phi_{m,n}$ is given by
\begin{eqnarray}
    \langle\bm{t}\rvert\hat{U}\lvert\bm{s}\rangle=\frac{\mathrm{perm}(U_{\bm{s},\bm{t}})}{\sqrt{s_1!\cdots s_m!t_1!\cdots t_m!}}.
\end{eqnarray}
In the above equation, $U_{\bm{s},\bm{t}}$ is an $n\times n$ matrix composed of multiple concatenated submatrices, where each of the $s_i\times t_j$ submatrices is obtained by repeating $u_{ij}$ for $s_it_j$ times. For example, when $n=4$ and $m=2$,
\begin{eqnarray}
    U_{(2,2),(1,3)}=
    \begin{pmatrix}
        u_{11}&u_{12}&u_{12}&u_{12}\\
        u_{11}&u_{12}&u_{12}&u_{12}\\
        u_{21}&u_{22}&u_{22}&u_{22}\\
        u_{21}&u_{22}&u_{22}&u_{22}\\
    \end{pmatrix}.
\end{eqnarray}
And $\mathrm{perm}(A)$ is the permanent of the matrix $A$, which is defined as
\begin{eqnarray}
    \mathrm{perm}(A)=\sum_{\sigma\in S_n}\prod_{i=1}^na_{i,\sigma(i)}.
\end{eqnarray}
Calculating the permanent of $U_{\bm{s},\bm{t}}$ for a Haar random matrix $U$ is conjectured to be an extremely difficult task for any classical computer. This conjecture forms the basis of certain claims of quantum advantage \cite{Aaronson2011,Harrow2017}.

The most time-consuming part of the classical simulation of an MPQNN is the computation of the probability distribution of the output photon partitions. To compute the entire probability distribution of the output photon partitions, we need to calculate the permanent of $U_{\bm{s},\bm{t}}$ for not a single partition $\bm{t}$, but all $n+m-1\choose m-1$ possible partitions. Thus, we designed a dynamic programming algorithm to calculate the amplitude, thereby accelerating the simulation.

If no photon enters the linear optical network (i.e. $n=0$), the amplitude $\langle\bm{t}\rvert\hat{U}\lvert\bm{s}\rangle$ is obviously $1$. Now we consider the case $n>0$, in which at least one of $s_1,\cdots,s_m$ is greater than $0$. Let $i$ be the smallest index that makes $s_i>0$. Utilizing the row expansion of the permanent, we can obtain the recursive formula of the amplitude as
\begin{eqnarray}
    \langle\bm{t}\rvert\hat{U}\lvert\bm{s}\rangle=\sum_{t_j>0}u_{ij}\sqrt\frac{t_j}{s_i}\langle\bm{t}_j^-\rvert\hat{U}\lvert\bm{s}_i^-\rangle,
\end{eqnarray}
where $\bm{s}_i^-=(s_1,\cdots,s_i-1,\cdots,s_m)\in\Phi_{m,n-1}$. We can calculate the amplitude $\langle\bm{t}\rvert\hat{U}\lvert\bm{s}\rangle$ for the partition $\bm{t}\in\Phi_{m,n}$ by calculating the amplitudes $\langle\bm{t}_j^-\rvert\hat{U}\lvert\bm{s}_i^-\rangle$ for at most $m$ partitions $\bm{t}_j^-\in\Phi_{m,n-1}$. 

In the algorithm, for each $\bm{t}^\prime$ satisfying $t_j^\prime\le t_j$ for $j\in[m]$, there exists exactly one $\bm{s}^\prime$ with which the function $\mathrm{AMP}(\bm{s}^\prime,\bm{t}^\prime)$ is called. And $\mathrm{AMP}(\bm{s}^\prime,\bm{t}^\prime)$ is called at most once to calculate the amplitude $\langle\bm{t}^\prime\rvert\hat{U}\lvert\bm{s}^\prime\rangle$. The number of all such $\bm{t}^\prime$ is
\begin{eqnarray}
    \prod_{j=1}^m(t_j+1)\le\left(\frac{n+m}m\right)^m.
\end{eqnarray}
The complexity of the algorithm is $O\left(m\left(\frac{n+m}m\right)^m\right)$.

\begin{algorithm}
\caption{calculation of amplitudes}
\SetAlgoLined
\KwIn{$n,m,\bm{s},\bm{t},U$}
\KwOut{$\langle\bm{t}\rvert\hat{U}\lvert\bm{s}\rangle$}
\SetKwFunction{Amp}{AMP}
\SetKwProg{Fn}{Function}{:}{}
\Fn{\Amp{$s,t$}}{
    \If{$amp[s,t]$ is defined}{
        \Return $amp[s,t]$\;
    }
    \eIf{$\sum_{i=1}^mt_i=0$}{
        $amp[s,t]\gets1$\;
        \Return $amp[s,t]$\;
    }{
        find the smallest index $i$ that makes $s_i>0$\;
        calculate $s_i^-$\;
        $sum\gets0$\;
        \For{$j\gets1$ to $m$}{
            \If{$t_j>0$}{
                calculate $t_j^-$\;
                $sum\gets sum+u_{ij}\sqrt\frac{t_j}{s_i}$\Amp{$s_i^-,t_j^-$}\;
            }
        }
        \Return $sum$
    }
}
\end{algorithm}

\bibliography{reference}

\end{document}